\documentclass{article}

\usepackage{microtype}
\usepackage{graphicx}
\usepackage{subfigure}
\usepackage{booktabs} 

\usepackage{hyperref}

\usepackage[accepted]{icml2024}

\usepackage{amsmath}
\usepackage{amssymb}
\usepackage{mathtools}
\usepackage{amsthm}
\usepackage{bm}
\usepackage{enumitem}
\usepackage{caption}

\usepackage[capitalize,noabbrev]{cleveref}

\usepackage{url}

\usepackage{color}

\newcommand\blue[1]{#1}

\usepackage{tikz}

\theoremstyle{plain}
\newtheorem{theorem}{Theorem}[section]
\newtheorem{proposition}[theorem]{Proposition}

\theoremstyle{definition}
\newtheorem{definition}[theorem]{Definition}

\theoremstyle{remark}

\usepackage[textsize=tiny]{todonotes}

\icmltitlerunning{
Bottleneck-Minimal Indexing for Generative Document Retrieval
}

\newcommand\doc{x}
\newcommand\query{q}
\newcommand\idx{t}

\newcommand\docv{\mathbf{\doc}}

\usepackage{mathrsfs}

\newcommand\dspace{\mathscr{X}}  

\newcommand\dvar{X}
\newcommand\qvar{Q}
\newcommand\idvar{T}

\newcommand\docs{\mathcal{X}}
\newcommand\queries{\mathcal{Q}}
\newcommand\ids{\mathcal{T}}

\newcommand\alphbet{V}

\usepackage[normalem]{ulem}

\graphicspath{{img/}}

\begin{document}

\twocolumn[
\icmltitle{
Bottleneck-Minimal Indexing for Generative Document Retrieval
}

\icmlsetsymbol{equal}{*}

\begin{icmlauthorlist}
\icmlauthor{Xin Du}{equal,wise}
\icmlauthor{Lixin Xiu}{equal,utokyo}
\icmlauthor{Kumiko Tanaka-Ishii}{fsci}
\end{icmlauthorlist}

\icmlaffiliation{fsci}{Department of Computer Science and Engineering, Waseda University}
\icmlaffiliation{wise}{Waseda Research Institute for Science and Engineering, Waseda University}
\icmlaffiliation{utokyo}{Department of Mathematical Informatics, The University of Tokyo}

\icmlcorrespondingauthor{Xin Du}{duxin@aoni.waseda.jp}
\icmlcorrespondingauthor{Kumiko Tanaka-Ishii}{kumiko@waseda.jp}

\icmlkeywords{Machine Learning, ICML}

\vskip 0.3in
]

\printAffiliationsAndNotice{\icmlEqualContribution}  

\begin{abstract}
We apply an information-theoretic perspective to reconsider generative
document retrieval (GDR), in which a document $\doc \in \docs$ is
indexed by $\idx \in \ids$, and a neural autoregressive model is
trained to map queries $\queries$ to $\ids$. GDR can be considered to
involve information transmission from documents $\docs$ to queries
$\queries$, with the requirement to transmit more bits via the indexes
$\ids$. By applying Shannon's rate-distortion theory, the optimality
of indexing can be analyzed in terms of the mutual information, and
the design of the indexes $\ids$ can then be regarded as a {\em
bottleneck} in GDR. After reformulating GDR from this perspective, we
empirically quantify the bottleneck underlying GDR. Finally, using the
NQ320K and MARCO datasets, we evaluate our proposed bottleneck-minimal
indexing method in comparison with various previous indexing methods,
and we show that it outperforms those methods.
\end{abstract}

\section{Introduction}

\begin{figure}[t]
\centering
\includegraphics[width=0.9\linewidth]{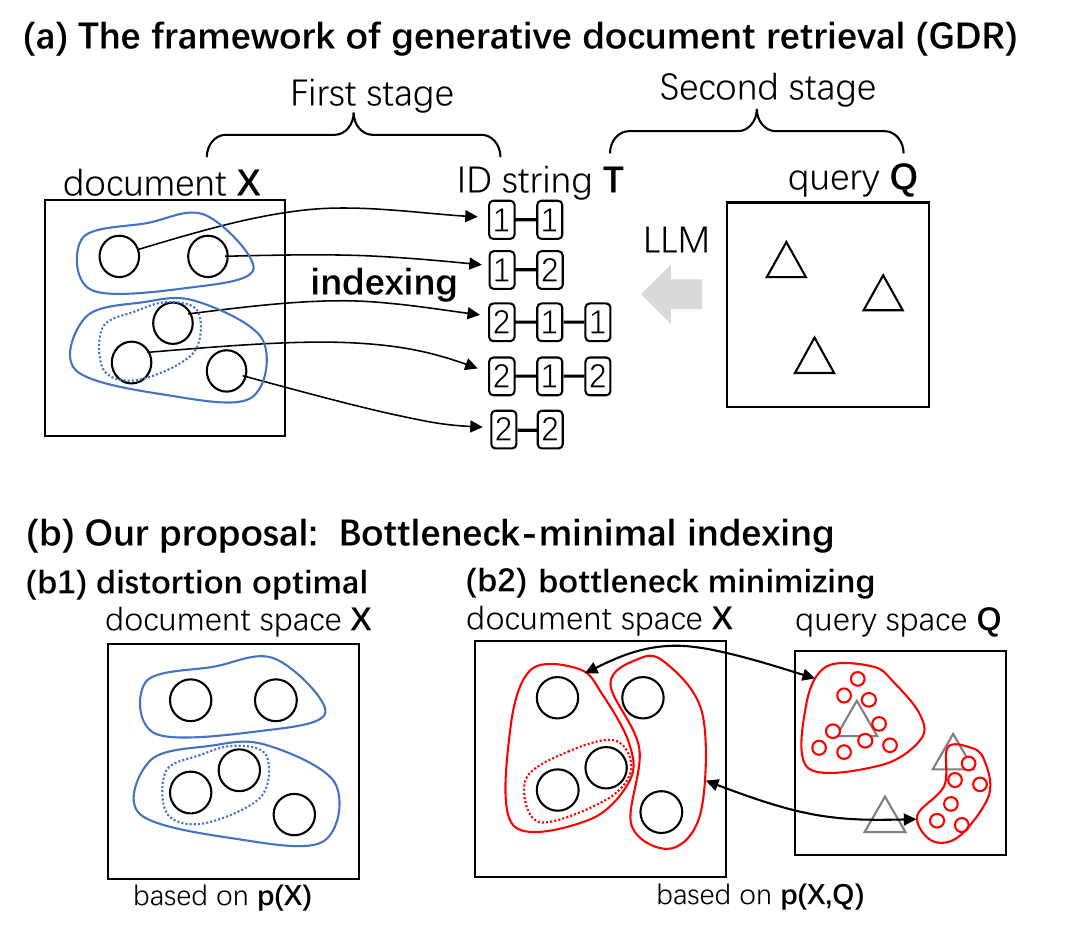}
\caption{(a) Generative document retrieval (GDR) framework. (b) Our contribution using bottleneck-minimal indexing: (b-1) distortion-optimal indexing for documents $\docs$; (b-2) optimal indexing for both documents and queries $\queries$.
\label{fig:overview}}
\vskip-1em
\end{figure}

The importance of accurate document retrieval is increasing,
especially with the limitation of recent large language models.
Generative document retrieval (GDR) is a new,
promising framework for information
retrieval. First, every document $\doc \in
\docs$ (a document set) is represented by a short, distinct
identifier string $\idx \in \ids$ (the identifier set). Then, an
autoregressive model, typically a neural network, is trained to map a
query $\query\in \queries$ (a query set) to an identifier string
$\idx$ representing a document $\doc$. Figure 1(a) schematically
illustrates GDR.

GDR thus typically involves two stages, and the main research question has been about how to acquire good identifier strings $\ids$ for documents $\docs$ in the first stage. Previously, \citet{tay2022dsi} clustered $\docs$ according to vectorial embeddings generated by the BERT model \citep{devlin2019bert}, by using the $k$-means algorithm. \citet{bevilacqua2022autoregressive} used a substring of each document as its ID string. All of those works were based on good intuition about how to semantically partition $\docs$, but we can theoretically reconsider how best to partition it in terms of $\ids$.

For such reconsideration, it is natural to apply the rate-distortion theory initiated by Shannon \citep{shannon1948mathematical,shannon59,berger68}, which is based on the mutual information. Let $\dvar$, $\idvar$, and $\qvar$ denote random variables defined over $\docs$, $\ids$, and $\queries$, respectively. As the rate-distortion theory only involves two terms corresponding to $\dvar$ and $\idvar$, the aim, in short, is to acquire $\idvar$ that minimally distorts $\dvar$. All previous GDR proposals thus consider its distortion optimality, as illustrated in Figure 1(b-1).

In this paper, we deal with the particularity of the second stage, which involves mapping $\qvar$ to $\idvar$. The design of $\idvar$, which we call {\em indexing} here, was previously considered only with respect to $\dvar$, but the best $\idvar$ can theoretically be obtained by considering both $\dvar$ and $\qvar$. For this purpose, we apply the information bottleneck theory by \citet{tishby2000ib}, an extension of rate-distortion theory.

Under this bottleneck theory, GDR involves information transmission from documents to queries, with the requirement to transmit more bits via $\idvar$. Obviously this suggests a tradeoff between the two relations $\idvar\leftrightarrow \dvar$ and $\idvar \leftrightarrow \qvar$, and $\idvar$ then becomes a {\em bottleneck} of $\qvar$ in searching $\dvar$. Our paper's main contribution is to provide a theoretical formalization to evaluate the quality of $\idvar$ in terms of not only $\dvar$ but also $\qvar$.

The bottleneck suggests that the optimal $\idvar$ is determined by the probability distribution of queries rather than documents, because $\qvar$ is posterior to $\dvar$. This suggests a novel design for $\idvar$ by clustering $\qvar$ rather than $\dvar$. We experimentally evaluate this approach's effectiveness in comparison with previous methods, including $k$-means \citep{hartigan1979algorithm} and locality-sensitive hashing \citep{datar2004locality}. The results show that our idea indeed improves the Recall@1 score on the NQ320K and MARCO Lite datasets, by 1.26 and 3.72 points, respectively, with a finetuned \texttt{T5-base} model. The margins increase greatly to 7.06 and 6.45 points, respectively, for \texttt{T5-mini}, which has fewer parameters.
The code is available at \url{https://github.com/kduxin/Bottleneck-Minimal-Indexing}.

\section{Related Works}
We start with an overview of the recent progress in GDR in Section
\ref{sec:rw-gdr}. In Section \ref{sec:rw-reprlearn}, we briefly
summarize methods in discrete representation learning, which provide
options for indexing methods in GDR. In Section
\ref{sec:rw-distortion}, we introduce the historic rate-distortion
theory initiated by Shannon and the bottleneck theory by
\citet{tishby2000ib}.

\subsection{Generative Document Retrieval}
\label{sec:rw-gdr}
Conventional document retrieval methods use a score function to
quantify the degree of relevance between a document and a query. Score
functions have been manually chosen and include the Hamming distance,
\citep{salakhutdinov2009semantic}, TF-IDF
\citep{manning1999foundations}, BM25 score
\citep{robertson2009probabilistic}, and cosine similarity between
vectoral embeddings of texts
\citep{mikolov2013distributed,devlin2019bert,karpukhin2020dense,khattab2020colbert}.
Alternatively, score functions can also be learned from data, e.g.,
based on a learning-to-rank loss function
\citep{burges2005learning,cao2007learning,li2023learning}.

In contrast, generative document retrieval aims to directly generate a
document $\doc$ from a query $\query$ by using a sequence-to-sequence
model. In relation to indexing, \citet{cao2021autoregressive} proposed
the concept of the entity retrieval task, in which a model generates
an {\em entity} comprising either words, phrases, or article titles.
While documents are typically long, with hundreds of words, an entity
is short, and their work showed the potential effectiveness of using
entities as indexes. Since then, several works have studied generation
of short textual summaries to represent long documents
\citep{bevilacqua2022autoregressive,chen2022corpusbrain}. Using short
texts as indexes was also explored for recommendation tasks
\citep{geng2022recommendation}.

\citet{tay2022dsi} generalized this idea as a {\em differentiable
search index} (DSI), a more straightforward indexing method to obtain
document ID strings that are acquired by applying hierarchical
$k$-means clustering to a set of document vectoral embeddings produced
by a BERT model \citep{devlin2019bert}. In DSI, a document set $\docs$
is partitioned into clusters that are organized hierarchically (by
having subclusters in each cluster). An identifier number represents
each cluster, and a document is identified by a sequence of
identifiers $\idx$ representing the hierarchical clusters. This simple
method was further explored in
\citet{wang2022nci,mehta2023dsi,tang2023semantic,zeng2023scalable,rajput2023recommender}.
Recent works \citep{sun2023learning,jin2023language} proposed to learn
the index $\ids$ by using an autoencoder model with $\docs$.

While all these previous works consider organization of $\docs$ to design $\ids$, our contribution is to also include $\queries$ in the designed approach.

\subsection{Discrete Representation Learning}
\label{sec:rw-reprlearn}
Acquisition of $\ids$ can be considered as one kind of representation learning. Learning of discrete representations has been studied in multiple domains, as summarized below.

\paragraph{Vector quantization (VQ).} VQ encodes vectors in a Euclidean space as discrete codes with the least distortion. \citet{lloyd1982least,wu2019vector} recognized that the optimal vector quantization corresponds to the $k$-means clustering algorithm's results when the distortion is defined via the Euclidean distance. A Euclidean space can be partitioned by assuming different structures, such as a tree \citep{nister2006scalable}, in which case the approach corresponds to hierarchical $k$-means clustering. Deep learning techniques can be combined to learn better vector quantization \citep{oord2017vqvae}.

\paragraph{Hash-based methods.} A hash function maps vectors to discrete classes. Locality-sensitive hashing (LSH) \citep{datar2004locality} uses a property of $p$-stable distributions to construct a family of hash functions that is guaranteed to map close vectors to the same class. By using multiple independent hash functions, the dissimilarity of two vectors is measured by the Hamming distance. Semantic hashing \citep{salakhutdinov2009semantic} learns hash functions by using an autoencoder to acquire compact discrete representations. Many works since have studied the use of deep neural networks as hash functions \citep{venkateswara2017deephn,chaidaroon2017variational,jin2019unsupervised,lin2022deep}.

\subsection{Rate-Distortion Theory}
\label{sec:rw-distortion}
Rate-distortion (RD) theory \citep{shannon1948mathematical} provides an information-theoretic framework for optimal indexing in GDR. RD theory is based on $I(\dvar; \idvar)$, the mutual information for transmitting data between $\dvar$ and $\idvar$. Such transmission introduces distortion that is measured by a distortion function $d(\dvar,\idvar) \mapsto\mathbb{R}_+$. This function depends on the application, and the choice is sometimes not obvious \citep{slonim2002information}. The overall distortion $\mathbb{E}_{\dvar,\idvar} d(\dvar,\idvar)$ has a tradeoff with the transmission rate $I(\dvar; \idvar)$. For $p(\idvar|\dvar)$, the probability of $\idvar$ given $\dvar$, let $p^*(\idvar|\dvar)$ be the probability when $\idvar$ is optimal for $\dvar$ via RD theory. In GDR, this $\idvar$ would be considered the optimal indexing of documents.

The vector quantization problem mentioned in Section
\ref{sec:rw-reprlearn} is closely related to RD theory
\citep{gray1984vector}. RD theory can be extended to include a third
variable $\qvar$ via the information bottleneck (IB) theory
\cite{tishby2000ib}, and the distortion function is implicitly
determined via the mutual information $I(\idvar;\qvar)$. Here,
$\idvar$ is called an {\em information bottleneck} (IB) because it
determines the ``data transmission rate'' between $\dvar$ and $\qvar$.

The IB theory was originally proposed for bottom-up hierarchical (i.e., agglomerative) clustering \citep{slonim1999agg} applied to words \citep{slonim2000document}. It has also been applied to explain the superior generalizability of feedforward deep neural networks \citep{tishby2015deep,saxe2019information}.

We adopt IB theory to analyze GDR. The relation between documents $\dvar$ and queries $\qvar$ can be complex, especially in cross-lingual \citep{grefenstette2012cross} or multimodal \citep{gabeur2020multi} scenarios. Hence, the distortion function $d(\dvar,\idvar)$ must involve the queries $\qvar$, which naturally leads to the IB theory.

\section{Bottleneck-Minimal Indexing}
\subsection{Mathematical Setting}

As mentioned above, in the GDR context, each document $\doc\in\docs$
is indexed with an ID string $\idx\in\ids$, described by a function
$f(\doc)=\idx$. It is anticipated that information loss occurs during
the application of $f$, as the semantic associations between documents
are simplified into discrete ID strings. To quantify this information
loss, the domain of $f$ must be considered within an abstract semantic
space of documents, denoted as $\dspace$. Here, a document is a point
in $\dspace$, and the set $\docs$ of all documents in a dataset is a
subset of $\dspace$, i.e., $\doc\in\docs\subset\dspace$. For
convenience, $\qvar$ is defined over the same semantic vector space
$\dspace$, and $\idvar$ is defined over $\ids$, the set of all ID
strings.

The indexing function $f$ is formally defined as $f:\dspace\to\ids$.
By leveraging $f^{-1}: \idx\mapsto \dspace_\idx
\subset \dspace$, the semantic space $\dspace$ can be split into
multiple regions. According to the standard GDR setting
\citep{tay2022dsi}, it is required that $|\ids|=|\docs|$; every
$\idx\in\ids$ is associated with a distinct document $\doc$ in
$\docs$, the only document in $f^{-1}(\idx)\cap\docs$. When
$|\ids| < |\docs|$, indicating that an ID string is associated with
multiple documents, this configuration typically aligns with
semantic-hash methods where hash collison is desirable
\citep{salakhutdinov2009semantic}.



\subsection{Distortion Optimality}
\label{sec:ib}
\label{sec:beta}
Previous works considered the design of $\ids$ only with respect to
$\docs$. As mentioned above, Shannon's rate-distortion theory can be
applied to acquire $\ids$ as an optimal split of $\dspace$,
formulated as:
\begin{equation}
    \begin{aligned}
        \min_{p(\idvar\mid\dvar)} ~~~& I(\dvar; \idvar), \\
    \end{aligned}
\end{equation}
where $I(\dvar;\idvar)$ denotes the mutual information of $\dvar$ and
$\idvar$. This formalization allows for a reconsideration of the
optimal $\idvar$ given $\dvar$ for previous works, as will be
analyzed in Section \ref{sec:comparison}.


Information retrieval involves another term, $\qvar$, and the optimization can be extended as follows:
\begin{equation}
    \begin{aligned}
        \min_{p(\idvar\mid\dvar)} ~~~& I(\dvar; \idvar) \\
        \text{s.t.} ~~~& I(\idvar; \qvar) \geq \varepsilon,
    \end{aligned}
    \label{eq:optim}
\end{equation}
where $\varepsilon>0$ is a threshold. This additional constraint is
essential, as $I(\dvar; \idvar)=0$ when $\idvar$ is constant, but it
also decreases $I(\idvar; \qvar)$.
The mutual information $I(\idvar;\qvar)$ is closely connected to the
likelihood $p(\idvar|\qvar)$ that reflects retrieval accuracy;
therefore, having $I(\idvar;\qvar)$ in the optimization problem
facilitates optimization of retrieval accuracy.

Formula (\ref{eq:optim}) characterizes the tradeoff between index
conciseness, denoted by $I(\dvar;\idvar)$, and retrieval accuracy,
linked to $I(\idvar;\qvar)$. This dual perspective highlights the
superiority of one indexing method over another: for a fixed
$I(\dvar;\idvar)$, aiming for a higher $I(\idvar;\qvar)$ improves
retrieval accuracy. Conversely, with a constant $I(\idvar;\qvar)$, the
emphasis shifts towards more compact indexes, encouraging the sharing
of semantics across the indexes.

This optimal split is equivalently reformulated via a Lagrangian
$\mathcal{L}$ as follows:
\begin{equation}
    \mathcal{L}(p(\idvar | \dvar)) =
    I(\dvar; \idvar) - \beta I(\idvar; \qvar),
    \label{eq:ib}
\end{equation}
where $\beta$ is the Lagrange multiplier. This coefficient $\beta$
determines the relative importance of $I(\idvar; \qvar)$ in
$\mathcal{L}$. When $\beta\to 0$, the latter term can be ignored, and
$T$ shrinks to a single point so that the former term decreases to 0.
Conversely, as $\beta$ becomes large, $\idvar$ tends to simply copy
$\dvar$. Because the possibilities for $\dvar$ are observed via the
subset $\docs$, a large $\beta$ corresponds to observing more
documents.

To summarize, $\beta$ is a hyperparameter related to the document
set's size, $|\docs|$. This is natural because the optimal indexing
$p^*(\idvar|\dvar)$ would change if the number of documents
increases. Nevertheless, some indexing method is optimal with respect
to different $\beta$ under certain assumptions, as we will show in
Section \ref{sec:comparison}.

\subsection{GDR Bottleneck}
\label{sec:bottleneck-gdr}
We now examine the tradeoff, or bottleneck, between $I(\dvar;\idvar)$
and $I(\idvar;\qvar)$. \citet{tishby2015deep} showed that the
Lagrangian in Formula (\ref{eq:ib}) is equivalent to the following:
\begin{equation}
    \mathcal{L}(p(\idvar | \dvar)) = I(\dvar; \idvar) + \beta I(\dvar; \qvar | \idvar) + \text{constant}.
    \label{eq:ib2}
\end{equation}
The equivalence requires assuming the Markovian relation $\idvar\leftrightarrow \dvar\leftrightarrow \qvar$, i.e., $p(\idvar | \dvar, \qvar) = p(\idvar | \dvar)$, and Appendix \ref{sec:mi-equality} summarizes the proof.

In this formulation, the conditional mutual information $I(\dvar; \qvar | \idvar)$ quantifies the information distortion due to $\idvar$. $I(\dvar; \qvar)$ is constant and unrelated to the indexing $\idvar$. Hence, a small $I(\dvar; \qvar | \idvar)$ means a good $\idvar$ for representing the joint distribution $p(\dvar,\qvar)$.

\begin{figure}[t]
    \centering
    \includegraphics[width=0.7\linewidth]{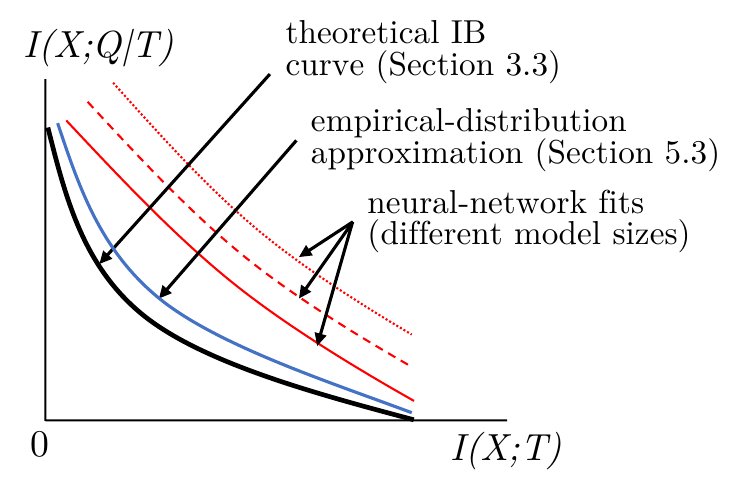}
\vskip-1.0em
 \caption{Bottleneck curves \label{fig:bottleneck-curve} }
 \vskip-1.0em
\end{figure}

As mentioned above, \citet{tishby2000ib} defined this tradeoff as an {\em information bottleneck}. The bottleneck can be visualized by a curve in the space of the two terms of Formula (\ref{eq:ib2}), as shown in Figure \ref{fig:bottleneck-curve}. The lower bound of $I(\dvar;\qvar|\idvar)$ captures the optimal indexing, and it monotonically decreases as $I(\dvar;\idvar)$ increases. We will empirically show in Section \ref{sec:bottleneck-experiment} that this bottleneck curve indeed occurs for GDR.

\subsection{A Theoretical Solution to Optimality}

\citet{tishby2000ib} showed that a stationarity point of the
Lagrangian $\mathcal{L}$ in Formula (\ref{eq:ib2}) must satisfy the
following equation:
\begin{equation}
    p^*(\idvar | \dvar) = \frac{p^*(\idvar)}{Z(\dvar,\beta)}
    \exp\Bigl(
        -\beta \text{KL}\bigl[
            p(\qvar | \dvar) \bigm\lVert p(\qvar | \idvar)
            \bigr]
        \Bigr),
    \label{eq:ibsolution}
\end{equation}
where $Z(\dvar,\beta)$ is a probability normalization term,
$p^*(\idvar)=\mathbb{E}_\dvar[p^*(\idvar|\dvar)]$. The
Kullback-Leibler divergence term reflects the information distortion
caused by $f:\doc\mapsto\idx$ in terms of the discrepancy between
$p(\qvar|\doc)$ and $p(\qvar|\idx)$.




In GDR, we pursue an indexing $f:\doc\mapsto\idx$ that is consistent
with this optimal solution $p^*(\idvar|\dvar)$. That is, the
identifier string $\idx=f(\doc)$ for a document $\doc$ should be
drawn from this ``best'' probability distribution
$p^*(\idvar|\dvar=\doc)$. A natural way to incorporate this intuition
is to evaluate $f$ by a likelihood function $p(\docs,\queries|f)$
defined as follows.

\begin{definition}
    $f:\docs\to\ids$ is called a {\em bottleneck-minimal indexing}
    (BMI) if it maximizes the likelihood function as follows:
    \begin{equation}
        p(\docs,\queries | f)
        \equiv\prod_{\doc\in\docs}
        p^*\bigl(\dvar=\doc \bigm| \idvar=f(\doc)\bigr).
        \label{eq:likelihood}
    \end{equation}
    where $p^*$ is given by Formula (\ref{eq:ibsolution}).
    \label{def:bmi}
\end{definition}

Definition \ref{def:bmi} is presented in a general manner,
intentionally avoiding assumptions about any specific distribution
family for $p(\qvar|\dvar)$ and $p(\qvar|\idvar)$ as outlined in
Formula (\ref{eq:ibsolution}). In Section \ref{sec:indexes}, we will
examine some existing indexing methods regarding their relations with
BMI. Under this background, a new indexing method will be introduced
in Section \ref{sec:method}.






\section{Indexing Methods}
\label{sec:comparison}
\label{sec:indexes}
In the following, we consider $\dspace$ as the semantic vector space produced by a pretrained BERT model. In this section we introduce different indexing methods under this perspective and compare them theoretically, while Section \ref{sec:compare} describes our experimental comparison.

An index $\idx \in \ids$ is represented as a sequence of elements of
{\em alphabet}, denoted as $\alphbet$. An ID string of length $m$ is
represented as $\idx=[\idx_1,\cdots,\idx_m]$, where
$\idx_i\in\alphbet$ is the string's $i$-th ``digit'' for
$i=1,\cdots,m$.

\subsection{Hierarchical Random Indexing (HRI)}
\label{sec:hri}
We consider the most basic method, random indexing. For each document,
an ID digit is randomly selected according to some prior distribution
$p(\idvar)$. $\docs$ is partitioned into $|V|$ subsets. Furthermore,
each subset is recursively partitioned for representation by $V$. The
recursive subdivision constitutes a {\em hierarchy} of subdivisions of
depth $m$. A special case arises when $|\alphbet|\geq|\docs|$ and
$m=1$; such indexing is termed {\em atomic} \citep{tay2022dsi}, since
$\ids$ lacks hierarchical structure.

A representative partition is obtained when $p(\qvar|\dvar)$ is uniformly distributed over a compact subset of $\dspace$. In this case, queries corresponding to a document $\doc$ are completely unrelated to the semantics of $\doc$.

In the experiments described in Section \ref{sec:compare}, we set $\alphbet=[1,2,\cdots,30]$. The ID string's maximum length $m$ was set to the minimum value such that $|\alphbet|^m \geq |\docs|$.

\subsection{Hierarchical $k$-Means Indexing (HKmI)}
\label{sec:hkmi}
The typical indexing uses a hierarchical $k$-means clustering algorithm. A document's ID string $\idx=[\idx_1,\cdots,\idx_m]$ is set to the indices of clusters of depth $m$, such that the document belongs at every level of the hierarchy, where $k$-means clustering is applied to partition the set of documents.

In a $k$-means clustering process, the clusters' centroid vectors are optimized to minimize the sum of the squared distances from each centroid to the cluster members. This procedure can be interpreted within a maximum-likelihood formulation by assuming that a document belonging to a cluster is sampled from a Gaussian distribution with the mean vector as the centroid.

Hence, $k$-means provides a BMI if we assume $p(\qvar|\dvar)\sim
N(\mu_\doc, \Sigma)$ and $p(\qvar|\idvar)\sim N(\mu_\idx, \sigma^2 I)$,
where $N$ represents a Gaussian distribution; $\mu_\doc$ and
$\mu_\idx$ are mean vectors specific to $\doc$ and $\idx$,
respectively; $\Sigma$ is any covariance matrix; and $\sigma^2 I$ is a
diagonal matrix. This conclusion immediately follows from substituting
the Gaussian density functions into Formula (\ref{eq:ibsolution}). The
optimality holds for any choice of $\beta$ in Formula (\ref{eq:ib}). A
proof is provided in Appendix \ref{sec:bmi-kmeans}.
The linkage between $k$-means clustering and information
bottleneck theory has been theoretically explored by
\citet{still2003geometric} in depth.

Following \citet{tay2022dsi}, we set $\mu_\doc$ as the vectoral embedding generated by a BERT model \citep{devlin2019bert} for a document $\doc$. In our experiments, we used the same settings for $\alphbet$ and $m$ as in HRI (Section \ref{sec:hri}). These settings were used as the default for the {\em neural corpus indexer} (NCI) \citep{wang2022nci}.

\subsection{Locality-Sensitive Hashing Indexing (LSHI)}
Locality-sensitive hashing (LSH) can be considered in the same
formulation. In LSH the index $\idx_i$ is Boolean; that is, $V =
\{0,1\}$, where each element of
$\idx=[\idx_1,\cdots,\idx_i,\cdots,\idx_m]$ is independently generated
by a $p$-stable LSH algorithm \citep{datar2004locality}, which is a
hyperplane classifier in Euclidean space. The hyperplane's location
and direction are randomly determined. Unlike hierarchical indexing
methods, LSH-based indexing does not produce ``semantic prefixes'';
nevertheless, $\ids$ encodes location-related information.

In our experiments, LSH indexing was implemented in three steps.
First, a standard LSH code, a Boolean vector, was acquired for every
document. Second, every fifth entry of the vector was mapped to
$\alphbet=[1,2,\cdots,32]$. Third, LSH code collisions were resolved
by appending additional digits acquired by a hierarchical $k$-means
algorithm.

\subsection{Our Proposal: Bottleneck-Minimal Indexing (BMI)}
\label{sec:method}

Formula (\ref{eq:ibsolution}) indicates that the optimal indexing is
dictated by the distributions $p(\qvar|\dvar)$ and
$p(\qvar|\idvar)$ over the query space, rather than the document space.
In other words, any indexing method implemented without considering
the distribution of $\qvar$ would not be able to acquire the optimal
indexing. Typical previous works on GDR have this problem.
The concept of leveraging query information has been
empirically investigated in the contexts of vector quantization
\citep{gupta2022bliss,lu2023knowledge} and recommendation systems
\citep{zeng2023scalable}; our contribution extends this exploration by
offering a theoretical rationale for the significance of query
distribution, grounded in information bottleneck theory.

Hence, we propose a new indexing method that incorporates both the
queries and the documents.
For the purposes of simplicity in this paper, we assume
$p(\qvar|\dvar = \doc)$ and $p(\qvar|\idvar = \idx)$ to follow
Gaussian distributions, enabling the analytical derivation of
bottleneck-minimizing indexing. In a manner akin to HKmI discussed in
Section \ref{sec:hkmi}, $k$-means emerges as the optimal indexing
strategy under this Gaussian assumption. However, our approach
diverges from HKmI by employing hierarchical $k$-means clustering on
the set $\{\mu_{\qvar|\doc}: \forall\doc\in\docs\}$, which comprises
the mean vectors of the Gaussian $p(\qvar|\dvar = \doc)$ for each
document $\doc$.
We examine $k$-means here for a fair comparison with HKmI, in addition
to its simplicity. However, Definition \ref{def:bmi} of BMI can be
used to analyze more complex clustering algorithms under more
sophisticated assumptions on $p(\qvar|\doc)$ and $p(\qvar|\idx)$.

Unlike $\mu_\doc$ for HKmI, however, it is not straightforward to
obtain $\mu_{\qvar|\doc}$, and it must be estimated for each document
$\doc$. In this paper, we estimate $\mu_{\qvar|\doc}$ as the mean of
the BERT embeddings of queries generated from document $\doc$,
which is the maximum-likelihood estimator given the
previously mentioned Gaussian distribution assumption of
$p(\qvar|\doc)$. A proof is provided in Appendix \ref{sec:mle}. We
chose BERT for consistency with previous works, including DSI
\citep{tay2022dsi} and NCI \citep{wang2022nci}; in a future work, we
may examine other models that are adapted to short texts.

Let $\queries_\doc$ denote a set of queries for $\doc$. The document
$\doc$ is now represented by mean(BERT($\queries_\doc$)), the mean
vector of the BERT embeddings of queries in $\queries_\doc$, instead
of the document's BERT embedding. The hierarchical $k$-means algorithm
is applied to
$\{\mu_{\qvar|\doc}=\text{mean}(\text{BERT}(\queries_\doc)):
\doc\in\docs \}$, which produces a new set of ID strings, $\ids$.

We follow \citet{wang2022nci} in constructing $\queries_\doc$ to
comprise the following three kinds of queries, as follows:
\begin{description}[noitemsep,topsep=0pt]
\item[\textbf{RealQ}:] Real queries from the training set;
\item[\textbf{GenQ}:] Queries generated by a pretrained model \texttt{docT5query} \footnote{
\scriptsize\url{huggingface.co/castorini/doc2query-t5-base-msmarco}
}\citep{nogueira2019doc2query};
\item[\textbf{DocSeg}:]Random segments of the original documents.
\end{description}
Section \ref{sec:data} gives the details of generating these queries.

In a real application, the distributions $p(\qvar|\dvar = \doc)$ and
$p(\qvar|\idvar = \idx)$ may deviate from Gaussian assumptions,
potentially exhibiting multi-modal characteristics. Under these
conditions, $k$-means clustering methods become suboptimal and should
be replaced by another sophisticated clustering methods that
reflect the true distribution of the data. Developing such methods
form a future direction.

\section{Data and Settings}
\subsection{Datasets and Metrics}
\label{sec:data}
We evaluated different indexing methods on two datasets: NQ320K \citep{kwiatkowski2019natural}, and MARCO Lite, which is a subset extracted from the document ranking dataset in MS MARCO \citep{nguyen2016msmarco}. The upper half of Table \ref{tbl:data} summarizes the basic statistics of the two datasets.

The entire MS MARCO dataset has 3.2 million documents and 367,013
queries. MARCO Lite was constructed by first randomly selecting half
the queries and then extracting those queries' gold-standard
documents.

The lower half of Table \ref{tbl:data} summarizes our proposed method's generated queries (i.e., $\queries_\doc$) for creating document ID strings, as described at the end of Section \ref{sec:method}. For NQ320K, we followed the settings in \citet{wang2022nci} to produce 15 \textbf{GenQ} queries; for MARCO Lite, we produced 5 \textbf{GenQ} queries per document. As for \textbf{DocSeg} queries, we randomly selected 10$\sim$12 segments (depending on the document length) per document as queries for both datasets; each segment had around 60 words.

\begin{table}[t]
    \centering
    \small
 \caption{Descriptive statistics of the datasets (upper) and generated queries (lower).}
 \vskip-0.8em
    \label{tbl:data}
    \setlength{\tabcolsep}{1pt}
    \begin{tabular*}{\linewidth}{@{\extracolsep{\fill}} lcccc}
        \toprule
         & \multicolumn{2}{c}{NQ320K} & \multicolumn{2}{c}{MS MARCO Lite} \\
         & \# & mean \# words & \# & mean \# words \\
        \midrule

        documents & 109,739 & 4902.7 & 138,457 & 1210.1 \\
        queries (train) & 307,373 & 9.2 & 183,947 & 6.0 \\
        queries (test) & 7,830 & 9.3 & 2,792 & 5.9 \\
        \midrule
        \multicolumn{5}{l}{\textbf{Generated queries}} \\
        GenQ  & 1,646,085 & 5.6 & 692,285 & 5.5 \\
        DocSeg & 1,168,585 & 62.9 & 1,393,329 & 59.0 \\
        \bottomrule
    \end{tabular*}
\end{table}

We followed previous works by using the recall (Rec@N) and mean
reciprocal rank (MRR) for evaluation, where a higher Rec@N or MRR
indicates a better GDR system.
For each query, GDR produces a ranking
of documents with respect to their degrees of relevance to the query,
via a beam-search procedure as suggested in \citet{tay2022dsi}.
In our evaluation across two datasets, each query was associated with
a single gold-standard document. The Rec@N measures the percentage of
queries for which the gold-standard document is among the top N
documents in the ranking, while the MRR is equal to the gold-standard
document's MRR. We used values of $N=1,10,100$.

\subsection{Model and Training Settings}
\label{sec:ablation}
We used the same neural-network architecture as in NCI
\citep{wang2022nci}. For sequence-to-sequence generation of an index
for a query, the NCI architecture combined a standard transformer
\citep{vaswani2017transformer} with a {\em prefix-aware
weight-adaptive} (PAWA) decoder. The Transformer encoder's parameters
were initialized with T5 weights \citep{raffel2020exploring} acquired
by large-scale pretraining, while the decoder weights were randomly
initialized.

We tested models of different sizes, which were initialized from the
weights of
T5-tiny\footnote{\scriptsize\url{https://huggingface.co/google/t5-efficient-tiny}},
T5-mini\footnote{\scriptsize\url{https://huggingface.co/google/t5-efficient-mini}},
T5-small\footnote{\scriptsize\url{https://huggingface.co/t5-small}},
and T5-base\footnote{\scriptsize\url{https://huggingface.co/t5-base}},
where the string after ``T5'' indicates the weights. The PAWA decoder
had four transformer layers. All models were trained using the default
hyperparameters of NCI, as provided in its official GitHub
repository\footnote{\scriptsize\url{github.com/solidsea98/Neural-Corpus-Indexer-NCI}}.

For ablation tests, we considered different $\queries_\doc$ variants.
Apart from setting
$\queries_\doc=\text{GenQ}+\text{RealQ}+\text{DocSeg}$ to include all
three kinds of queries, we also tested $\text{GenQ}$ alone and
$\text{GenQ} + \text{RealQ}$. In addition, the previous method of
clustering documents by their embeddings also constitutes an ablated
version, denoted here as \textbf{Doc}, because it is almost equivalent
to $\text{DocSeg}$, which uses an average of random document segments.

\subsection{Estimation of Mutual Information}
\label{sec:estimation}
To verify the information bottleneck in GDR, the mutual information $I(\dvar;\idvar)$ and $I(\dvar;\qvar|\idvar)$ must be calculated. They are defined as follows:
\begin{align}
    I(\dvar;\idvar) &= \mathbb{E}_{\dvar, \idvar} \log \frac{p(\dvar|\idvar)}{p(\dvar)}, \\
    I(\dvar;\qvar|\idvar) &= \mathbb{E}_{\dvar,\idvar,\qvar}
    \log \frac{p(\dvar,\qvar|\idvar)}{p(\dvar|\idvar) p(\qvar|\idvar)} \label{eq:cmi} \\
    &= \mathbb{E}_{\dvar,\idvar,\qvar} \log \frac{p(\dvar|\qvar)}{p(\idvar|\qvar)}\frac{p(\idvar)}{p(\dvar)}.
    \label{eq:mi-dqid}
\end{align}
Appendix \ref{sec:mi-equality} gives the detailed derivation of
Formula (\ref{eq:mi-dqid}). The two definitions above are calculated
as further detailed in Appendix \ref{sec:ixt}. To acquire the mutual
information values for $|\ids|<|\docs|$, we considered a generalized
GDR task of predicting the prefixes of ID strings rather than the
whole strings, for a prefix length $l$. In addition to the default
setting $l=m$ (i.e., standard GDR), we also tested $l=2,3$.

\section{Experiments}
\subsection{Quantification of GDR Bottleneck}
\label{sec:bottleneck-experiment}

\begin{figure}[t]
    \centering
    \begin{minipage}[t]{0.48\linewidth}
        \centering
        \begin{tikzpicture}
            \draw (0, 0) node[inner sep=0]
            {\includegraphics[width=\linewidth]{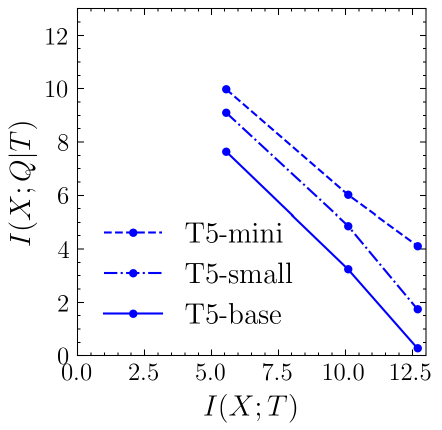}};
            \draw (-0.16\linewidth, 0.16\linewidth) node {\large (a)};
        \end{tikzpicture}
    \end{minipage}
    \begin{minipage}[t]{0.48\linewidth}
        \centering
        \begin{tikzpicture}
            \draw (0, 0) node[inner sep=0]
            {\includegraphics[width=\linewidth]{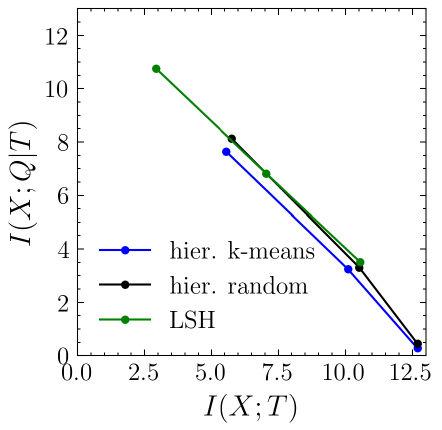}};
            \draw (-0.16\linewidth, 0.16\linewidth) node {\large (b)};
        \end{tikzpicture}
    \end{minipage}
\caption{Experimental information bottleneck curves corresponding to Figure \ref{fig:bottleneck-curve}. (a) Empirical information curves for HKmI, measured with T5 models of different sizes. (b) Empirical information curves for different indexing methods, estimated on the NQ320K dataset with the T5-base model. \label{fig:ibplane}}
\vskip-1.5em
\end{figure}

\begin{figure}
    \centering
    \includegraphics[width=0.60\linewidth]{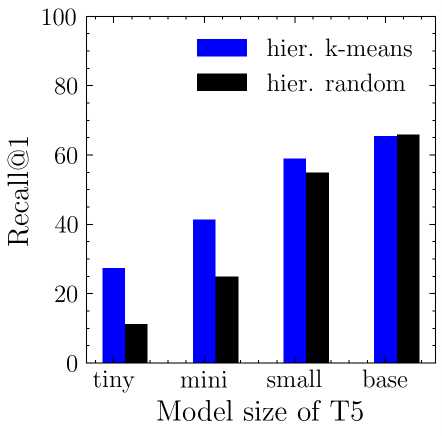}
\caption{Rec@1 scores on the test set of NQ320K for document IDs
generated by hierarchical $k$-means (blue) or random (red) clustering,
for different sizes of the finetuned language model T5.}
    \label{fig:kmeans-vs-random}
\end{figure}

\begin{table*}[htbp]
    \centering
    \small
\caption{Performance of different indexing methods.}
\vskip-1em
    \label{tbl:performance}
    \begin{tabular*}{\linewidth}{@{\extracolsep{\fill}} l|cccc|cccc}
        \toprule
         & \multicolumn{4}{c|}{NQ320k}
         & \multicolumn{4}{c}{MARCO Lite} \\
        Indexing Methods & Rec@1 & Rec@10 & Rec@100 & MRR@100
        & Rec@1 & Rec@10 & Rec@100 & MRR@100
        \\
        \midrule
        \multicolumn{1}{l|}{\textbf{T5-mini}} & & & & & & & & \\
        Hier. random clustering & 24.90 & 54.89 & 76.78 & 34.72 & 5.98 & 10.57 & 19.59 & 7.69\\
        Locality-sensitive hashing & 27.94 & 57.28 & 78.07 & 37.63 & 7.06 & 20.42 & 41.26 & 11.28 \\
        Hier. $k$-means clustering & 41.43 & 72.16 & 87.75 & 52.30 & 7.09 & 23.17 & 50.50 & 12.51 \\
        \midrule
        Our method (BMI) & \textbf{48.49} & \textbf{76.73} & \textbf{88.84} & \textbf{58.48} & \textbf{13.54} & \textbf{42.77} & \textbf{71.49} & \textbf{22.94} \\
        \midrule\midrule
        \multicolumn{1}{l|}{\textbf{T5-small}} & & & & & & & & \\
        Hier. random clustering & 54.93 & 77.20 & 87.79 & 62.75 & 12.64 & 43.88 & 71.91 & 22.42 \\
        Locality-sensitive hashing & 52.16 & 77.01 & 87.87 & 61.05 & 15.62 & 47.46 & 76.00 & 25.71 \\
        Hier. $k$-means clustering & \textbf{58.94} & \textbf{82.94} & \textbf{92.08} & \textbf{67.67} & 22.35 & 56.30 & 83.02 & 33.04 \\
        \midrule
        Our method (BMI) & 58.61 & 82.76 & 91.42 & 67.38 & \textbf{26.71} & \textbf{61.00} & \textbf{85.46} & \textbf{38.02} \\
        \midrule\midrule
        \multicolumn{1}{l|}{\textbf{T5-base}} & & & & & & & & \\
        Hier. random clustering & 65.85 & 84.16 & 91.51 & 72.59 & 39.00 & 70.45 & 88.65 & 49.54 \\
        Locality-sensitive hashing & 62.82 & 83.27 & 91.42 & 70.37 & 39.83 & 70.77 & 87.00 & 50.29 \\
        Hier. $k$-means clustering & 65.43 & 85.20 & 92.64 & 72.73 & 41.48 & 74.39 & 89.65 & 52.57 \\
        \midrule
        Our method (BMI) & \textbf{66.69} & \textbf{86.17} & \textbf{93.23} & \textbf{73.91} & \textbf{45.20} & \textbf{76.04} &\textbf{90.62} & \textbf{55.47} \\
        \bottomrule
    \end{tabular*}
\end{table*}

Figure \ref{fig:ibplane}(a) shows $I(\dvar; \qvar | \idvar)$ with respect to $I(\dvar; \idvar)$, as measured on the NQ320K dataset's training set. The ID strings were acquired by HKmI. A curve was obtained for every neural network model of a different size. The points on each curve were acquired with different $l$ values, as detailed at the end of Section \ref{sec:estimation}. As seen in the figure, a larger model corresponded to a curve closer to the graph's lower-left corner, thus suggesting the existence of a bottleneck as anticipated by the IB theory.

Figure \ref{fig:ibplane}(b) compares several existing indexing methods. Here, HKmI was closer to the lower-left corner than the other two methods. Hence, the optimal condition for $k$-means indexing better fits the reality of data.


Next, Figure \ref{fig:kmeans-vs-random} compares the random (black)
and $k$-means (blue) hierarchical indexing methods in terms of the
Rec@1 on the NQ320K test set. As the model size is reduced, Rec@1
shows a general decline for both methods. With the reduction in size,
the superiority of the $k$-means method over the random method became
apparent. While the performance of the two methods was similar for the
T5-base model, the Rec@1 of the $k$-means method was more than twice
that of the random method in the T5-tiny model.

These observations suggests that the IB curve indicates the quality of
GDR, specifically in terms of its distance from the graph's lower-left
corner. While the Rec@1 merely evaluates the case of $|\ids|=|\docs|$,
the IB curve evaluates the case of $|\ids| < |\docs|$, thus enabling
recognition of overfitting in GDR.

\subsection{Comparison Among Indexing Methods}
\label{sec:compare}
We also evaluated our proposed indexing method in comparison with the existing methods. Table \ref{tbl:performance} summarizes our experimental results on the two datasets. The table's three blocks correspond to training models of different sizes: T5-mini, T5-small, and T5-base, from top to bottom. Each row represents an indexing method, and each column corresponds to a metric. The scores were acquired on the test sets. A higher score indicates a better indexing method.

As seen in the table, our method (bottom row in each block) achieved the best scores under most settings. In particular, our method consistently outperformed hierarchical $k$-means (third row in each block), often with a large difference. On MARCO Lite and with T5-base (right side, bottom block), our method had a Rec@1 of 45.20, 3.72 points higher than the score for the original hierarchical $k$-means indexing. This margin was even larger than that between $k$-means and random clustering (2.48 points).

With the smaller model T5-mini (top block), the improvement achieved by our method was dramatic. Specifically, compared with the original hierarchical $k$-means, our method improved the Rec@1 by 7.06 points (17\%) on NQ320K and 6.45 points (91\%) on MARCO Lite.
\blue{
This improvement is corroborated by the observation that our method's
IB curve is located closer to the lower-left corner, as shown in
Figure \ref{fig:ibplane-mini-marco} in Appendix
\ref{sec:ib-plane-extra}.
}

These large improvements indicate that there is a better indexing than
with the previous methods, but this indexing must involve $\qvar$.
Recall that our method is also based on hierarchical $k$-means, and
the only difference from the original method is application of the
clustering algorithm to queries rather than documents. These
improvement results indicate that the information bottleneck theory
applies well to GDR. In other words, it is important to pursue
``bottleneck-minimal'' rather than ``distortion-minimal'' indexing.

\subsection{Ablation Study}

\begin{table}[h]
    \centering
    \small
\caption{Ablation study results on the test set of MARCO Lite.
GQ and RQ are abbreviations of GenQ and RealQ, respectively.}
\vskip-1em
    \label{tbl:ablation}
    \setlength{\tabcolsep}{1pt}
    \begin{tabular*}{\linewidth}{@{\extracolsep{\fill}} lccccc}
        \toprule
         & Rec@1 & Rec@10 & Rec@100 & MRR@100 \\
        \midrule
        \multicolumn{5}{c}{\textbf{T5-mini}} \\
        Doc & 7.09 & 23.17 & 50.50 & 12.51 \\
        GenQ & 9.24 & 27.51 & 56.27 & 15.55 \\
        GenQ+RealQ & 10.03 & 30.91 & 56.59 & 16.70 \\
        GQ+RQ+DocSeg & \textbf{13.54} & \textbf{42.77} & \textbf{71.49} & \textbf{22.94} \\
        \midrule
        \multicolumn{5}{c}{\textbf{T5-small}} \\
        Doc & 22.35 & 56.30 & 83.02 & 33.04 \\
        GenQ & 22.74 & 52.90 & 79.55 & 32.29 \\
        GenQ+RealQ & 25.25 & 58.92 & 82.84 & 36.52 \\
        GQ+RQ+DocSeg & \textbf{26.71} & \textbf{61.00} & \textbf{85.46} & \textbf{38.02}  \\
        \midrule
        \multicolumn{5}{c}{\textbf{T5-base}} \\
        Doc & 41.48 & 74.39 & 89.65 & 52.57 \\
        GenQ & 39.04 & 70.70 & 88.93 & 49.88 \\
        GenQ+RealQ & 39.43 & 72.31 & 88.93 & 50.62 \\
        GQ+RQ+DocSeg & \textbf{45.20} & \textbf{76.04} &\textbf{90.62} & \textbf{55.47} \\
        \bottomrule
    \end{tabular*}
\end{table}

Furthermore, we compared our proposed method with ablated versions of
it, as detailed in Section \ref{sec:ablation}. Table
\ref{tbl:ablation} summarizes the results on the MARCO Lite dataset.
Each row besides the last one represents an ablated version used a
certain subset of the generated queries as $\queries_\doc$ to estimate
$\mu_{\qvar|\doc}$, as mentioned in Section \ref{sec:method}.

With the smaller models (T5-mini and T5-small), even the ablated
versions outperformed Doc. For example, GenQ achieved a Rec@1 of 9.24
with T5-mini, which was higher than the Doc result of 7.09. This is
quite surprising because GenQ (or GenQ+RealQ) only used short queries.
As seen in the lower half of Table \ref{tbl:data}, the GenQ queries
had a mean length of only 5.5 words, much shorter than the document
lengths.

Nevertheless, there was a significant margin between the results with and without DocSeg. With T5-base, GenQ+RealQ+DocSeg had a Rec@1 of 45.20, 5.77 points higher than that for GenQ+RealQ.

The addition of DocSeg is effective for adding rich information existing in documents. Compared with the DocSeg queries, the GenQ and RealQ queries had limited variety: many of them were duplicates that compromised the retrieval performance.

\subsection{Comparison with SOTA}

\begin{figure*}[htbp]

\begin{minipage}[t]{.40\linewidth}
    \vspace{0pt}
    \centering
    \small
    \includegraphics[width=0.8\linewidth]{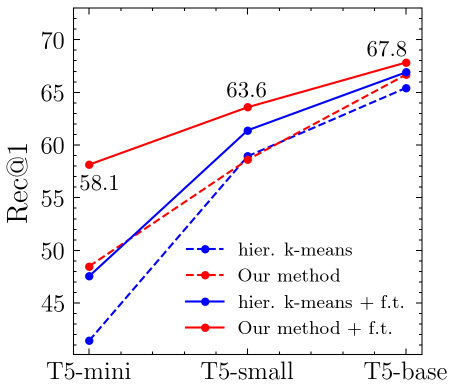}
    \caption{\blue{Performance enhancement
    by finetuning the \texttt{docT5query} model for generating
    improved GenQ documents.}}
    \label{fig:finetune}
\end{minipage}
\hspace{0.03\linewidth}
\begin{minipage}[t]{.55\linewidth}
    \vspace{0pt}
    \centering
    \small
\captionof{table}{
\blue{
Performance of BMI on NQ320K, with improved GenQ queries by finetuning
\texttt{docT5query} on the training set of NQ320K, compared with
existing GDR methods. }}
    \label{tbl:sota}
    \setlength{\tabcolsep}{1pt}
    \begin{tabular*}{\linewidth}{@{\extracolsep{\fill}} lccc}
        \toprule
        GDR Method & \multicolumn{2}{c}{Index Type} & Rec@1 \\
        \midrule
        DSI \citep{tay2022dsi} & static & number & 27.4 \\
        Ultron-PQ \citep{zhou2022ultron} & static & number & 25.6 \\
        NCI \citep{wang2022nci} & static & number & 66.9 \\
        Tied-At. \citep{nguyen2023generative} & static & number & 65.3 \\
        SEAL \citep{bevilacqua2022autoregressive} & static & textual & 26.3 \\
        Ultron-URL \citep{zhou2022ultron} & static & textual & 33.8 \\
        Our method & static & number & \textbf{67.8} \\
        \midrule
        LMIndexer \citep{jin2023language} & learned & textual & 66.3 \\
        GENRET \citep{sun2023learning} & learned & textual & 68.1 \\
        NOVO \citep{wang2023novo} & learned & textual & \textbf{69.3} \\
        \bottomrule
    \end{tabular*}
\end{minipage}
\end{figure*}

\blue{
Finally, we compare our proposed method against state-of-the-art
(SOTA) methods beyond the DSI/NCI framework utilizing the NQ320K
dataset. To show the power of our method at full strength, we
enhanced the GenQ queries by finetuning the document-to-query model
(\texttt{docT5query}) on the training set of NQ320K, thus producing
higher-quality ID strings. Details of this finetuning process are
left to Appendix \ref{sec:finetuning}.

Figure \ref{fig:finetune} shows the results of the finetuning process.
Solid curves shows the performance using enhanced GenQ queries, while
dashed curves represent the baseline results reported in the second
column of Table \ref{tbl:performance}. The HKmI method was also
examined for its use of enhanced GenQ queries as augmented training
data, although it, unlike our method, did not utilize these enhanced
GenQ queries for indexing.

An improvement in Rec@1 is observable across all model sizes for both
our method (in red) and HKmI (in blue), with our approach showing a
more significant enhancement, particularly with smaller models such as
T5-mini or T5-small. After fine-tuning, our method surpassed HKmI
across all model sizes. This substantial enhancement underscores the
critical role of accurately estimating $\mu_{\qvar|\doc}$, the mean of
query vectors, in improving the quality of document indexes.

Table \ref{tbl:sota} compares our new results with previous methods
beyond the DSI/NCI scope, using the T5-base retrieval model
for consistency with most prior studies.
Excluding NCI (i.e., HKmI) and our method---whose results are shown
in Figure \ref{fig:finetune}---the performance scores for all other
methods are taken directly from their original papers.

The methods are categorized based on whether indexes are determined
before or during the training phase of the retrieval model, labeled as
``static'' and ``learned,'' respectively, as detailed in Table
\ref{tbl:sota}'s second column. Moreover, the third column specifies
whether an index is a numeric string of text (a sequence of words).
Designing text-based indexing is inherently more complex and extends
beyond number-based methods. This analysis includes only methods which
reported their performance on the NQ320K dataset, which means we focus
on document retrieval. Related works in vector quantization or
recommendation systems are not compared here.

Our approach achieved a Rec@1 of 67.8, surpassing all models that
relied on static indexes. This performance closely competes with SOTA
methods, such as GENRET \citep{sun2023learning} at 68.1 and NOVO
\citep{wang2023novo} at 69.3, which employ textual ID strings and
design sophisticated ways of learning indexes concurrent with the
retrieval model's training.

These findings underscore our method's effectiveness and simplicity,
utilizing static, number-based indexes generated through a simple
$k$-means process. Our model's training procedure is the same as that
of NCI, without necessitating additional computational resources, with
the sole distinction being the production of document indexes by
incorporating query distribution. }

\section{Conclusion}
We introduced a new way to formulate generative document retrieval (GDR) by applying the information bottleneck theory. In GDR, a document $\doc \in \docs$ is indexed by $\idx \in \ids$, and queries $\queries$ are trained by neural autoregressive models for mapping to indexes $\ids$. GDR can thus be treated as the transmission of information from documents $\docs$ to queries $\queries$, with the requirement to transmit more bits via the indexes $\ids$. While previous methods considered how to acquire the best $\ids$ with respect to $\docs$, this work also incorporated $\queries$.

This reformulation naturally suggested a new document indexing method
to design GDR by using $\queries$. Our method, which applies the
hierarchical $k$-means algorithm to queries instead of document
vectors, achieved a significant improvement on two datasets in terms
of the recall and MRR scores. The margin was especially large when a
smaller neural network was used, with an improvement of up to 6.45
points, or 91\% relative improvement.

\blue{
In a comparison with various indexing methods, our method not only
outperformed all the other models that are based on static indexes,
but also achieved a competitive accuracy against those SOTA models
which employ textual ID strings and sophisticated index-learning
strategies during the retrieval model's training. These results
affirm the effectiveness of our approach, positioning it as a simple
yet strong competitor within the field. }

\section*{Acknowledgements}
This work was supported by JST CREST Grant Number JPMJCR2114.

\section*{Impact Statement}
This paper presents work whose goal is to advance the field of
machine learning and, specifically, generative information retrieval.
There are many potential societal consequences of our work, none of
which we feel must be specifically highlighted here.


\bibliography{main}
\bibliographystyle{icml2024}

\newpage
\appendix
\onecolumn

\section{Mathematical Details}
\label{sec:proof}

\subsection{$I(\dvar; \qvar|\idvar)$ in Information Bottleneck Context}
\label{sec:mi-equality}
The information bottleneck theory assumes the Markovian condition $\idvar\leftrightarrow\dvar\leftrightarrow\qvar$, which means that $p(\qvar|\dvar,\idvar)=p(\qvar|\dvar)$.

\begin{proposition}
\begin{equation}
    I(\dvar; \qvar|\idvar)=-I(\idvar; \qvar)+\text{constant}
    \label{eq:mi-equality}
\end{equation}
\end{proposition}
\begin{proof}
\begin{align}
    I(\dvar;\qvar|\idvar)
    &= \mathbb{E}_{\dvar,\idvar,\qvar}
    \log \frac{p(\dvar,\qvar|\idvar)}{p(\dvar|\idvar) p(\qvar|\idvar)} \label{eq:cmi2} \\
    &= \mathbb{E}_{\dvar,\idvar,\qvar}
    \log \frac{p(\qvar|\dvar,\idvar)p(\dvar|\idvar)}{p(\dvar|\idvar) p(\qvar|\idvar)} \\
    &= \mathbb{E}_{\dvar,\idvar,\qvar} \log \frac{p(\qvar|\dvar)}{p(\qvar|\idvar)} & \Bigl(\because p(\qvar|\dvar,\idvar)=p(\qvar|\dvar)\Bigr) \label{eq:cmi-simp} \\
    &= \mathbb{E}_{\dvar,\idvar,\qvar} \log \frac{p(\dvar|\qvar)}{p(\idvar|\qvar)}\frac{p(\idvar)}{p(\dvar)} & (\because \text{Bayes' theorem}) \\
    &= - \mathbb{E}_{\dvar,\idvar,\qvar} \log \frac{p(\idvar|\qvar)}{p(\idvar)}
    + \mathbb{E}_{\dvar,\idvar,\qvar} \frac{p(\dvar|\qvar)}{p(\dvar)} \\
    &= - I(\idvar; \qvar) + I(\dvar; \qvar).
\end{align}
The mutual information $I(\dvar; \qvar)$ is a constant determined by the dataset. Therefore, $I(\dvar; \qvar|\idvar)=-I(\idvar; \qvar)+$constant.
\end{proof}

\begin{proposition}
The Lagrangians in Formulas (\ref{eq:ib}) and (\ref{eq:ib2}) are equivalent.
\end{proposition}
\begin{proof}
Substitution of Formula (\ref{eq:mi-equality}) into (\ref{eq:ib})
immediately produces Formula (\ref{eq:ib2}), which implies the two
Lagrangians' equivalence.
\end{proof}

\subsection{$k$-Means Produces Bottleneck Minimum}
\label{sec:bmi-kmeans}

In both the HKmI (Section \ref{sec:hkmi}) and BMI (Section
\ref{sec:method}), the $k$-means clustering method is applied
hierarchically to the vector representations of documents $\docs$. We
demonstrate that at each hierarchical level, as $\docs$ is assigned
identifiers from the alphabet $\alphbet$, $k$-means clustering
effectively minimizes the information bottleneck as defined in
Definition \ref{def:bmi}.

Let $\docv$ represent the vector representation of a document $\doc$.
The goal of $k$-means clustering is to partition $\{\docv:
\doc\in\docs\}$ into $k$ clusters, aiming to minimize the sum of the
squared distances from each document to its cluster's centroid. This
is represented by the partition function $f:\doc\mapsto i$. The
optimization problem is described as follows \citep{hartigan1979algorithm}:
\begin{align}
    \min_{f} S(f) &= \sum_{i=1}^k
    \sum_{\doc\in f^{-1}(i)} \lVert \docv - \bar{\docv}_i\rVert^2,
    \label{eq:kmeans} \\
    \text{where}~\bar{\docv}_i
    &=\frac{1}{|f^{-1}(i)|}\sum_{\doc\in f^{-1}(i)} \docv.
    \label{eq:centroid}
\end{align}
This paper focuses on the optimal solution $f^*$ of Formula
(\ref{eq:kmeans}). We acknowledge various implementations for
achieving $f^*$; however, they are not covered in this paper.

We now aim to establish the following proposition:
\begin{proposition}
Assume that $p(\qvar|\doc)\sim N(\docv, \Sigma)$ and
$p(\qvar|\idx)\sim N(\mu_\idx, \sigma^2 I)$, where $\Sigma$ is an
arbitrary covariance matrix and $\sigma^2 I$ is an isotropic diagonal
matrix; $\mu_\idx$ is the vector estimated for every element
$\idx\in\alphbet$ in the alphabet $\alphbet$, Then, any optimal
solution $f^*$ to Formula (\ref{eq:kmeans}) also maximizes the
likelihood $p(\docs,\queries|f)$ as defined in Formula
(\ref{eq:likelihood}).
\label{thm:kmeans}
\end{proposition}

\begin{proof}

By definition, $p(\docs,\queries|f)=
\prod_{\doc\in\docs} p^*(\dvar=\doc\mid\idvar=f(\doc))$,
where $p^*$ is detailed in Formula (\ref{eq:ibsolution}).
From Bayes' theorem:
\begin{align}
    p^*(\dvar|\idvar)=\frac{p^*(\idvar|\dvar)}{p^*(\idvar)} p(\dvar),
    \label{eq:bayes}
\end{align}
where $p(\dvar)$ is a constant independent of $f$.

Applying Formula (\ref{eq:bayes}) and then inserting Formula
(\ref{eq:ibsolution}), we derive:
\begin{align}
    p(\docs,\queries|f)
    &\propto \prod_{\doc\in\docs}
    \frac{p^*(\idvar=f(\doc)|\dvar=\doc)}{p^*(\idvar=f(\doc))} \\
    &= \frac{1}{\prod_{\doc\in\docs} Z(\doc,\beta)}
    \exp\left(
        -\beta \sum_{\doc\in\docs} \text{KL} \Bigl[
        p(\qvar|\dvar=\doc) \Bigm\lVert p\bigl(\qvar|\idvar=f(\doc)\bigr)
        \Bigr]
    \right)  \label{eq:simp1}
\end{align}

Assuming $p(\qvar|\doc)\sim N(\docv, \Sigma)$ and
$p(\qvar|\idx)\sim N(\mu_\idx, \sigma^2 I)$, the KL divergence is
computed analytically; expanding the expression 
in Formula (\ref{eq:simp1}) leads to:
\begin{align}
    p(\docs,\queries|f) &= p(\docs,\queries|f, \{\mu_\idx: \idx\in\alphbet\}) \\
    &\propto 
    \exp\left\{
        -\frac{\beta}{2\sigma^2}
        \sum_{\doc\in\docs} \lVert \docv - \mu_{f(\doc)} \rVert^2
    \right\}
    \underbrace{
    \exp\left\{
        -\frac{\beta}{2} |\docs| \Bigl(
            d \log\sigma^2 - d - \log |\Sigma|
            + \text{tr}\bigl(\Sigma\bigr) / \sigma^2
        \Bigr)
    \right\},
    }_\text{constant} \\
    &\propto 
    \exp\left\{
        -\frac{\beta}{2\sigma^2}
        \sum_{\doc\in\docs} \lVert \docv - \mu_{f(\doc)} \rVert^2
    \right\} \\
    &=
    \exp\left\{
        -\frac{\beta}{2\sigma^2}
        \sum_{\idx\in\alphbet} \sum_{\doc\in f^{-1}(\idx)}
        \lVert \docv - \mu_\idx \rVert^2
    \right\}.
    \label{eq:simp2}
\end{align}

Since $\beta$ and $\sigma$ are constants, maximizing
$p(\docs,\queries|f)$ equates to minimizing
$\sum_{\idx\in\alphbet} \sum_{\doc\in f^{-1}(\idx)}
\lVert \docv - \mu_\idx \rVert^2$.
At the maximum, $\mu_\idx$ corresponds to the centroid
$\bar{\docv}_\idx$. This directly aligns with the $k$-means objective
in Formula (\ref{eq:kmeans}) and the selection of centroid as the
cluster center in Formula (\ref{eq:centroid}).

Therefore, any optimal solution $f^*$ to Formula (\ref{eq:kmeans})
also optimally maximizes the likelihood $p(\docs,\queries|f)$.

\end{proof}

The effectiveness of $k$-means in HKmI or BMI is validated through
Proposition \ref{thm:kmeans}. In HKmI, assuming $p(\qvar|\doc)\sim
N(\mu_\doc, \Sigma)$, where $\mu_\doc$ represents the BERT vector
representation for the document, aligns with maximizing
$p(\docs,\queries|f)$. For BMI, where the assumption is
$p(\qvar|\doc)\sim N(\mu_{\qvar|\doc}, \Sigma)$, we also achieve
maximization of $p(\docs,\queries|f)$.

However, it is a simplification in HKmI to assume $p(\qvar|\doc)\sim
N(\mu_\doc, \Sigma)$ as $\mu_\doc$ describes the document, not
necessarily its query distribution. In contrast, BMI's assumption of
$p(\qvar|\doc)\sim N(\mu_{\qvar|\doc}, \Sigma)$ is more suitable.

\subsection{
Average of Query Vectors as A Maximum-Likelihood Estimator for
Gaussian $p(\qvar|\doc)$ Mean}
\label{sec:mle}

\blue{
In Section \ref{sec:method}, we used the following approximation
\begin{equation}
    \mu_{\qvar|\doc} \approx \text{mean}(\text{BERT}(\queries_\doc)).
\end{equation}
Here, $\mu_{\qvar|\doc}$ represents the mean vector of the
distribution $p(\qvar|\doc)$ which we assumed to be Gaussian.
The expression
$\text{BERT}(\queries_\doc)=\{\text{BERT}(\query):
\query\in\queries_\doc\}$ denotes the set of BERT-generated
query vector embeddings associated with the gold-standard document
$\doc$. This section elaborates on the rationale behind the
approximation, demonstrating that
$\text{mean}(\text{BERT}(\queries_\doc))$ serves as a
maximum-likelihood estimator for $\mu_{\qvar|\doc}$.

To demonstrate that $\text{mean}(\text{BERT}(\queries_\doc))$ is a
maximum-likelihood estimator for $\mu_{\qvar|\doc}$, it is necessary
to assume that the queries in $\queries_\doc$ represent independent
samples drawn from $p(\qvar|\doc)$. $p(\qvar|\doc)$ is a Gaussian
distribution with mean $\mu_{\qvar|\doc}$ and covariance matrix
$\Sigma$ that is unknown. Then, the likelihood function
$p(\queries_\doc|\mu_{\qvar|\doc})$ can be formulated:
\begin{align}
p(\queries_\doc|\mu)
&= \prod_{\query\in\queries_\doc} p(\query|\mu) \\
&\propto \exp \left(
    -\frac{1}{2}
    \sum_{\query\in\queries_\doc}
    (\query - \mu)^\top \Sigma^{-1}
    (\query - \mu)
\right) \\
&=\exp -\frac{1}{2} \Biggl(
    \underbrace{
    |\queries_\doc| \mu^\top\Sigma^{-1}\mu
    - 2\mu^\top\Sigma^{-1}\sum_{\query\in\queries_\doc} \query 
    }_{=:g(\mu)}
    + \sum_{\query\in\queries_\doc} \query^\top \Sigma^{-1}\query
\Biggr),
\label{eq:gmu}
\end{align}
where $|\queries_\doc|$ denotes the number of queries. Here, $\query$
represents the BERT-generated vector embedding of a query.

$g(\mu)$ in Formula (\ref{eq:gmu}) is a quadratic function, and its
derivative is as follows:
\begin{equation}
\frac{\mathrm{d}g(\mu)}{\mathrm{d}\mu}
= 2\Sigma^{-1} \Biggl(
    |\queries_\doc|\mu - \sum_{\query\in\queries_\doc} \query
\Biggr)
\end{equation}

Upon solving $\mathrm{d}g(\mu)/\mathrm{d}\mu=0$, we derive the
maximum-likelihood estimator for $\mu_{\qvar|\doc}$ as:
\begin{equation}
    \hat{\mu}_{\qvar|\doc} =
    \frac{1}{|\queries_\doc|} \sum_{\query\in\queries_\doc} \query,
\end{equation}
which precisely corresponds to $\text{mean}(\text{BERT}(\query))$.
This outcome affirms that taking the average of BERT embeddings for
the set of queries $\queries_\doc$ associated with a document $\doc$
provide the most probable estimation of $\mu_{\qvar|\doc}$, under the
assumed Gaussian distribution of queries conditioned on documents.

Therefore, the accuracy of this approximation relies on the degree to
which the queries in $\queries_\doc$ adhere to the independent and
identically distributed (i.i.d.) assumption underlying
$p(\qvar|\doc)$. As demonstrated in Appendix \ref{sec:finetuning},
optimizing the GenQ queries through fine-tuning the document-to-query
model (which estimates $p(\qvar|\doc)$) results in enhanced retrieval
performance. This underscores the practical relevance of our
theoretical framework in actual applications.
}

\subsection{Estimation of $I(X;T)$ and $I(X;Q|T)$}
\label{sec:ixt}
Among the three variables, $\idvar$ is a discrete variable, while $\dvar$ and $\qvar$ are continuous variables in Euclidean spaces and do not follow standard distributions. In this paper, we estimate mutual information values via the empirical distributions $\hat{\dvar}$ and $\hat{\qvar}$, which are uniform over $\docs$ and $\queries$, respectively.

In other words,
\begin{align}
    I(\dvar;\idvar) &= \lim_{|\docs|\to\infty}
    \frac{1}{|\docs|} \sum_{\hat{\dvar}\in\docs} \log\frac{p(\hat{\dvar}|\hat{\idvar})}{1/|\docs|}, \\
    I(\dvar;\qvar|\idvar)
    &=\lim_{|\docs|\to\infty} \frac{1}{|\docs|}
    \sum_{\hat{\dvar}\in\docs}
    \sum_{\hat{\qvar}\in\queries}
    \log \frac{p(\hat{\dvar}|\hat{\qvar})}{p(\hat{\idvar}|\hat{\qvar})}
    \frac{1 / |\ids|}{1 / |\docs|}.
\end{align}
Here, $p(\hat{\dvar}|\hat{\idvar})$ is determined by the indexing. In contrast, $p(\hat{\dvar}|\hat{\qvar})$ is unrelated to the indexing and is a degenerate distribution: it has a probability mass of 1 if and only if $\hat{\dvar}$ is the gold-standard document for query $\qvar$. We estimate $p(\hat{\idvar}|\hat{\qvar})$ with a transformer neural network after fitting to the dataset.

The estimation of $p(\hat{\dvar}|\hat{\idvar})$ is not straightforward. In standard GDR, a document is assigned a distinct ID string, and $p(\hat{\dvar}|\hat{\idvar})$ is a degenerate distribution. In this case, $I(\dvar;\idvar)$ reaches its maximum, i.e., $H(\dvar)=\log |\docs|$, which corresponds to the rightmost points of the IB curves shown in Figure \ref{fig:bottleneck-curve}, where $\beta=\infty$.

To obtain the rest of the IB curves ($\beta<\infty$) in Figure
\ref{fig:bottleneck-curve}, we used a generalized GDR task that
permitted ID collisions. This was implementing by training the
transformer neural network to predict prefixes of a document ID string
$[\idx_1,\cdots,\idx_{l}]$, instead of the whole string
$[\idx_1,\cdots,\idx_{m}]$, where $l<m$. An ID prefix thus corresponds
to multiple documents. In this setting,
\begin{equation}
    p(\hat{\dvar}|\hat{\idvar}=\idx) = \left\{\begin{matrix}
        1/|\docs_\idx^l| & \hat{x}\in \docs_\idx^l \\
        0 & \text{otherwise}
    \end{matrix}\right.,
\end{equation}
where $\docs_\idx^l\subset\docs$ is the set of documents for which the
ID prefix is identical to the first $l$ digits of $\idx$.

\vfill

\section{More Results on Information Bottleneck Curves}
\label{sec:ib-plane-extra}

\begin{figure}[h]
\centering
\includegraphics[width=0.35\linewidth]{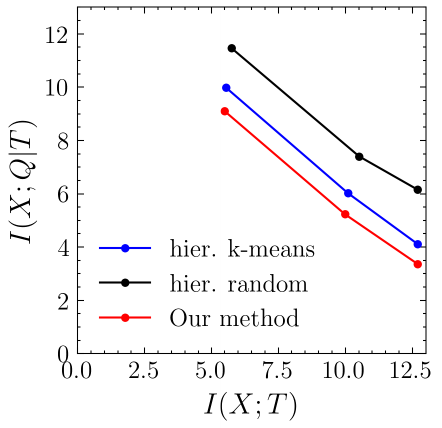}
\caption{ Information bottleneck curves for different indexing methods
    on the NQ320K training set with the T5-mini model. }
\label{fig:ibplane-mini-marco}
\end{figure}

Figure \ref{fig:ibplane-mini-marco} compares our proposed method (red)
with two other indexing methods on the information bottleneck plane,
as measured with the T5-mini model on the NQ320K training set. Our
method's curve is consistently located to the lower left of the other
two methods' curves. This suggests that our BMI method reduces the
bottleneck.

\vfill

\section{Details of Finetuning \texttt{docT5query}}
\label{sec:finetuning}

\blue{
The \texttt{docT5query} model underwent fine-tuning using the standard
next-word prediction task and a cross-entropy loss function.
Specifically, given a query $\query$ and its corresponding
gold-standard document $\doc$, the model parameters were adjusted to
maximize the likelihood $p(\query|\doc)=\prod_{i=1}^{|\query|-1}
p(\query_i|\doc;\query_1,\cdots,\query_{i-1})$, where $\query_i$
denotes the $i$-th word in the query text, and $|\query|$ is the query
length. Updates to the parameters were implemented using the AdamW
optimizer \citep{loshchilov2017decoupled}, with $\beta_1=0.9$,
$\beta_2=0.999$, $\text{eps}=10^{-8}$, and a weight decay of 0.01. The
learning rate was set at $5\times 10^{-5}$. This finetuning process
was executed 10 epochs on the training set of NQ320K.
}

\vfill

\section{Visualized Comparison of $\mu_\doc$ and
$\mu_{\qvar|\doc}$}
\label{sec:dimreduce}

\begin{figure}[h]
    \centering
    \begin{minipage}[t]{0.35\linewidth}
        \centering
        \includegraphics[width=\linewidth]{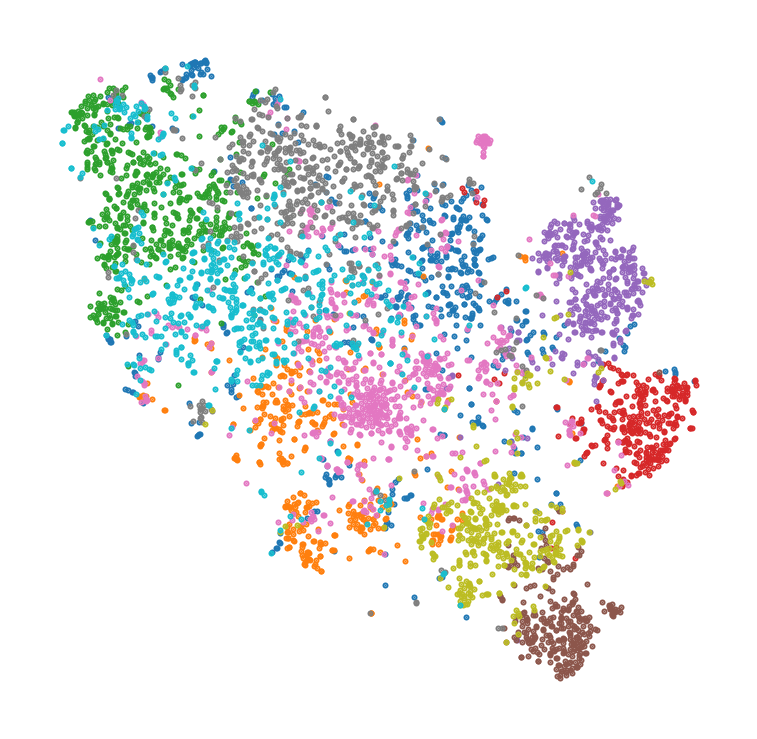}
 (a) Doc
    \end{minipage}
    \begin{minipage}[t]{0.35\linewidth}
        \centering
        \includegraphics[width=\linewidth]{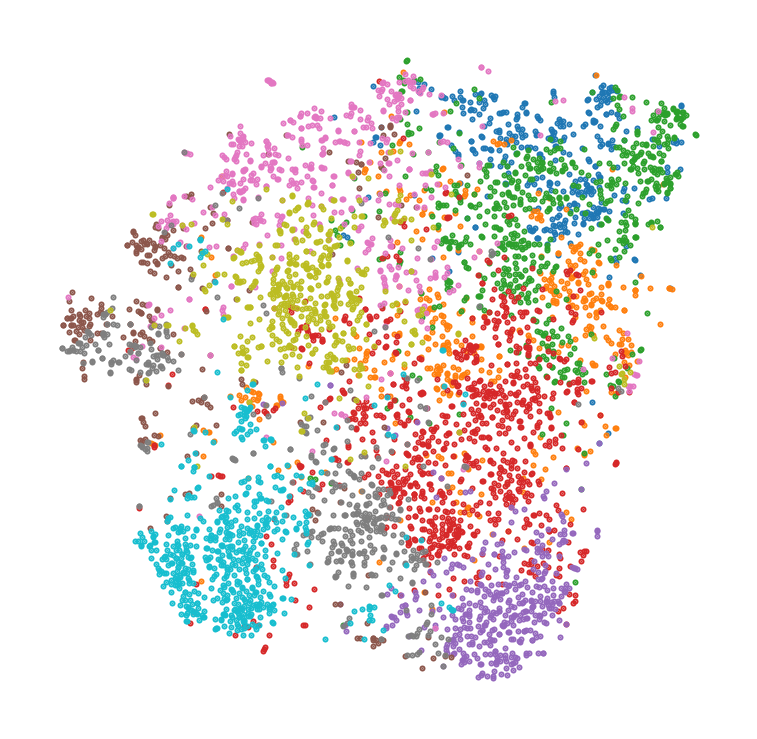}
 (b) GenQ
    \end{minipage} \\
    \begin{minipage}[t]{0.35\linewidth}
        \centering
        \includegraphics[width=\linewidth]{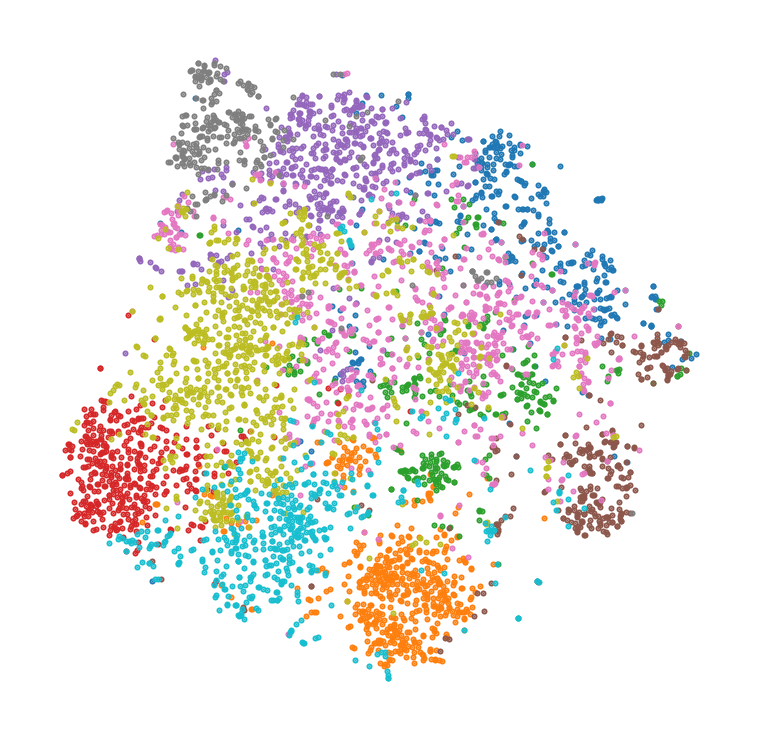}
 (c) GenQ+RealQ
    \end{minipage}
    \begin{minipage}[t]{0.35\linewidth}
        \centering
        \includegraphics[width=\linewidth]{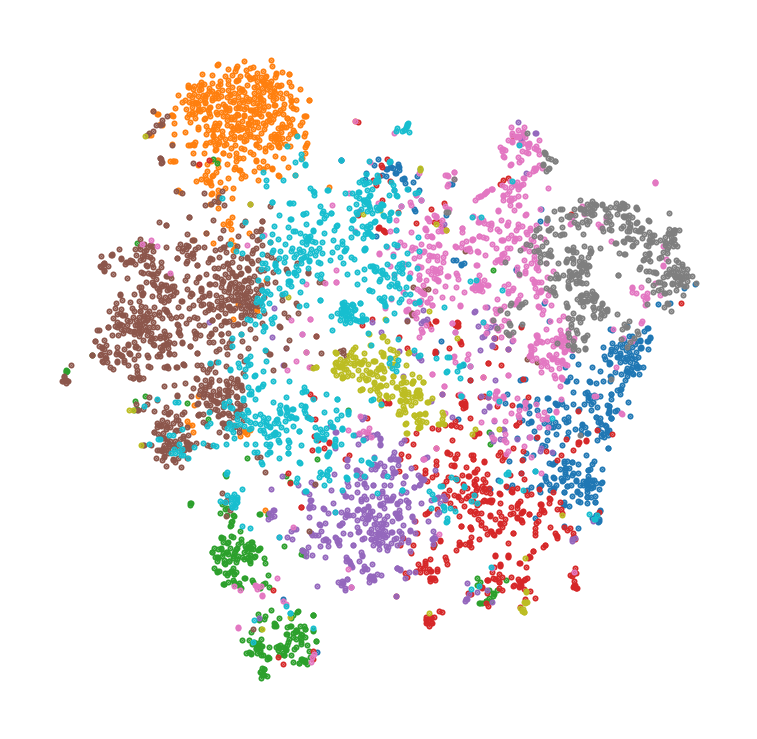}
 (d) GenQ+RealQ+DocSeg
    \end{minipage} \\
    \begin{minipage}[t]{0.35\linewidth}
        \centering
        \includegraphics[width=\linewidth]{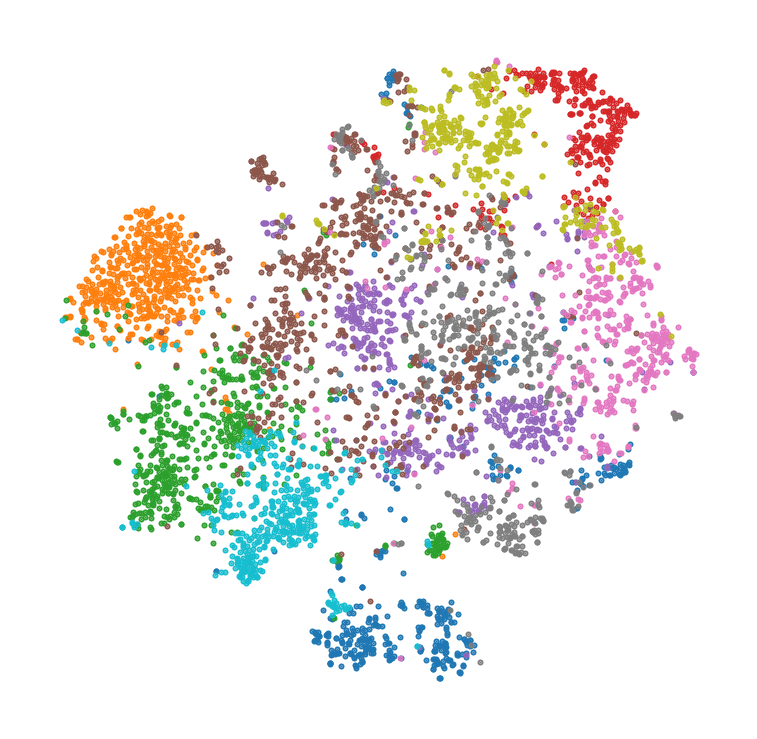}
 (e) GenQ (finetuned) + RealQ + DocSeg
    \end{minipage}
\caption{Visualization of (a) $\mu_{\doc}$ and (b-d)
$\mu_{\qvar|\doc}$ obtained with 5,000 documents in the NQ320K
dataset. The colors indicate different clusters assigned by $k$-means
clustering in the original high-dimensional space ($k=10$). Every
point represents a document and was obtained as (a) the document's
BERT embedding, (b) the mean BERT embedding of GenQ queries, (c) the
mean BERT embedding of GenQ+RealQ queries, (d\blue{-e}) the mean BERT
embedding of all queries (GenQ+RealQ+DocSeg). \blue{The GenQ queries
used to produce (e) were generated with a finetuned version of the
model on NQ320K, as detailed in Appendix \ref{sec:finetuning},
whereas those for (b-d) were not finetuned.} }
\label{fig:mu}
\end{figure}

Figure \ref{fig:mu} shows a visualization, using the t-SNE method
\citep{maaten2008tsne}, of the vectors on which we applied the
hierarchical $k$-means algorithm to obtain the indexing. 5,000
documents were randomly selected and are each represented by a point
in the graphs. Figures \ref{fig:mu}(a-d) respectively correspond to
the four rows in each block of Table \ref{tbl:ablation}. (a) was
obtained by using the documents' BERT embeddings, i.e., $\{\mu_\doc:
\doc\in\docs\}$, and (b-d) were generated by using the mean vectors of
queries, i.e.,
$\{\mu_{\qvar|\doc}=\text{mean}(\text{BERT}(\queries_\doc)):
\doc\in\docs\}$ where $\queries_\doc$ consists of GenQ queries,
GenQ+RealQ queries, and all queries (i.e., GenQ+RealQ+DocSeg) for
(b-d), respectively. \blue{In (e), we report the result of (d) when
GenQ queries were produced with the finetuned \texttt{docT5query}
model using the training set of NQ320K.}

Documents that were categorized in the same cluster are indicated by
the same color. The cluster borders are more evident in (d\blue{-e}) than in
(a-c). This visualization suggests that $\mu_{\qvar|\doc}$ better
represents ``semantics'' in the ID prefixes than $\mu_{\doc}$ does,
which explains why our method achieved better retrieval performance.

\end{document}